\documentclass[sigconf]{acmart}
\usepackage{booktabs} 
\usepackage{subcaption}
\usepackage{mathtools}
\usepackage{algorithm, algpseudocode}
\usepackage[rmdefault=ntxmi,sfdefault=iwona,scaled=1.0]{isomath}

\newcommand*\vs[1]{\vectorsym{#1}}

\newcommand*\mf{\mathop{}\!\mathtt{MF}}

\newtheorem{assumption}{Assumption}

\setcopyright{none}

\acmConference[The Web Conference 2019]{ACM WWW conference}{May 2019}{San Francisco, USA}
\acmYear{2019}
\copyrightyear{2019}

\begin{document}
\title{Efficient Ridesharing Order Dispatching with \\ Mean Field Multi-Agent Reinforcement Learning}

\author{
Minne Li\textsuperscript{1},
Zhiwei (Tony) Qin\textsuperscript{2},
Yan Jiao\textsuperscript{2},
Yaodong Yang\textsuperscript{1},
Zhichen Gong\textsuperscript{1},
Jun Wang\textsuperscript{1}, \\
Chenxi Wang\textsuperscript{3},
Guobin Wu\textsuperscript{3},
Jieping Ye\textsuperscript{3}}
\affiliation{%
  \institution{
  \textsuperscript{1}University College London,
  \textsuperscript{2}DiDi Research America,
  \textsuperscript{3}DiDi Research}
}

\renewcommand{\shortauthors}{M. Li et al.}
\renewcommand{\shorttitle}{Efficient Ridesharing Order Dispatching with MFRL}
\newcommand{\yaodong}[1]{{\bf \color{red} [[Yaodong: #1']]}}

\begin{abstract}
A fundamental question in any peer-to-peer ridesharing system is how to, both effectively and efficiently, dispatch user's ride requests to the right driver in real time.
Traditional rule-based solutions usually work on a simplified problem setting, which requires a sophisticated hand-crafted weight design for either centralized authority control or decentralized multi-agent scheduling systems.
Although recent approaches have used reinforcement learning to provide centralized combinatorial optimization algorithms with informative weight values, their single-agent setting can hardly model the complex interactions between drivers and orders.
In this paper, we address the order dispatching problem using multi-agent reinforcement learning (MARL), which follows the distributed nature of the peer-to-peer ridesharing problem and possesses the ability to capture the stochastic demand-supply dynamics in large-scale ridesharing scenarios.
Being more reliable than centralized approaches, our proposed MARL solutions could also support fully distributed execution through recent advances in the Internet of Vehicles (IoV) and the Vehicle-to-Network (V2N).
Furthermore, we adopt the mean field approximation to simplify the local interactions by taking an average action among neighborhoods. 
The mean field approximation is capable of globally capturing dynamic demand-supply variations by propagating many local interactions between agents and the environment.
Our extensive experiments have shown the significant improvements of MARL order dispatching algorithms over
several strong baselines on the gross merchandise volume (GMV), and order response rate measures.
Besides, the simulated experiments with real data have also justified that our solution can alleviate the supply-demand gap during the rush hours, thus possessing the capability of reducing traffic congestion.
\end{abstract}

%
%



\maketitle

\newcommand{\UU}{{\mathcal{U}}}

\newcommand{\vect}[1]{\mathbf{#1}}
\newcommand{\dotproduct}[2]{\transp{#1} {#2}}
\renewcommand{\exp}[1]{\text{exp}({#1})}
\newcommand{\dd}{\mathrm{d}}
\newcommand{\Expect}{\mathbb{E}}

\newcommand{\Eb}[2]{\ifthenelse{\equal{#1}{}}{\mbox{E}_{b} \left[ #2\right]}{\mbox{E}_{b} \left[ #2  \middle\vert #1 \right]}}

\newcommand{\St}{\mathcal{S}}
\newcommand{\Ob}{\mathcal{O}}
\newcommand{\Ac}{\mathcal{A}}
\newcommand{\expR}{\mathcal{R}}
\newcommand{\obj}{J}

\newcommand{\cp}{\vect{v}}
\newcommand{\ap}{\theta}
\newcommand{\x}{\vect{x}}
\newcommand{\e}{\vect{e}}
\newcommand{\w}{\vect{w}}
\newcommand{\sv}{\alpha_v}
\newcommand{\sw}{\alpha_w}
\newcommand{\Ncp}{{N_{\cp}}}
\newcommand{\Nap}{{N_{\ap}}}

\newcommand{\Gret}{R^{\lambda}}
\newcommand{\deltaret}{\delta^{\lambda}}

\newcommand{\comb}{\vect{z}}
\newcommand{\vfa}{{\hat{V}}}
\newcommand{\comp}{\psi}
\newcommand{\qfa}{{\hat{Q}}}

\newcommand{\gradt}[2]{{\bf g}({#2})}
\newcommand{\gradobj}{\grad{\obj(\ap)}{\ap}}
\newcommand{\gradtobj}{{\bf g}(\ap)}
\newcommand{\conva}{\tilde{\mathcal{Z}}}
\newcommand{\convhat}{\hat{\mathcal{Z}}}
\newcommand{\grada}{\gradt{\obj}{\ap}}
\newcommand{\gradtrue}{\grad{\obj}{\ap}}

\newcommand{\TODO}[1]{\textbf{TODO} $[$ \emph{#1} $]$}

\newcommand{\seeerrata}{\footnote{See errata in section~\ref{errata}}}

\newcommand{\opac}{Off-PAC\xspace}

\newcommand{\minne}[1]{{\bf \color{red} [[Minne says: #1']]}}
\newcommand{\jun}[1]{{\bf \color{blue} [[Jun says ``#1'']]}}
\newcommand{\tony}[1]{{\bf \color{violet} [[Tony says ``#1'']]}}
\newcommand{\yan}[1]{{\bf \color{cyan} [[Yan says ``#1'']]}}
\newcommand{\chenxi}[1]{{\bf \color{magenta} [[Chenxi says ``#1'']]}}
\newcommand{\zhichen}[1]{{\bf \color{red} [[Zhichen says ``#1'']]}}
\renewcommand{\thefootnote}{\fnsymbol{footnote}}

\section{Introduction}
\label{sec:intro}
Real-time ridesharing refers to the task of helping to arrange one-time shared rides on very short notice \cite{amey2011real, furuhata2013ridesharing}. Such technique, embedded in popular platforms including Uber, Lyft, and DiDi Chuxing, has greatly transformed the way people travel nowadays.
By exploiting the data of individual trajectories in both space and time dimensions, it offers more efficiency on traffic management, and the traffic congestion can be further alleviated as well \cite{92d45ce9b7c8482185638b0242fdfd52}.

One of the critical problems in large-scale real-time ridesharing systems is how to dispatch orders, i.e., to assign orders to a set of active drivers on a real-time basis.
Since the quality of order dispatching will directly affect the utility of transportation capacity, the amount of service income, and the level of customer satisfaction,
 therefore, solving the problem of order dispatching is the key to any successful ride-sharing platform.
In this paper, our goal is to develop an intelligent decision system to maximize the gross merchandise volume (GMV), i.e., the value of all the orders served in a single day, with the ability to scale up to a large number of drivers and robust to potential hardware or connectivity failures.

The challenge of order dispatching is to find an optimal trade-off between the short-term and long-term rewards.
When the number of available orders is larger than that of active drivers within the order broadcasting area (shown as the grey shadow area in the center of Fig. \ref{subfig:neighborhood}), the problem turns into finding an optimal order choice for each driver. 
Taking an order with a higher price will contribute to the immediate income; however, it might also harm the GMV in the long run if this order takes the driver to a sparsely populated area.
As illustrated in Fig. \ref{subfig:orderdispatching}, considering the two orders starting from the same area but to two different destinations, a neighboring central business district (CBD) and a distant suburb. A driver taking the latter one may have a higher one-off order price due to the longer travel distance, but the subsequent suburban area with little demand could also prevent the driver from further sustaining income.
Dispatching too many such orders will, therefore, reduce the number of orders taken and harm the long-term GMV. The problem becomes more serious particularly during the peak hours when the situation in places with the imbalance between vehicle supply and order demand gets worse.
As such, an intelligent order dispatching system should be designed to not only assign orders with high prices to the drivers, but also to anticipate the future demand-supply gap and distribute the imbalance among different destinations.
In the meantime, the pick-up distance should also be minimized, as the drivers will not get paid during the pick-up process; on the other hand, long waiting time will affect the customer experience.

\begin{figure*}
  \centering
	\begin{subfigure}[b]{.66\linewidth}
		\centering
		\includegraphics[height=1.2in]{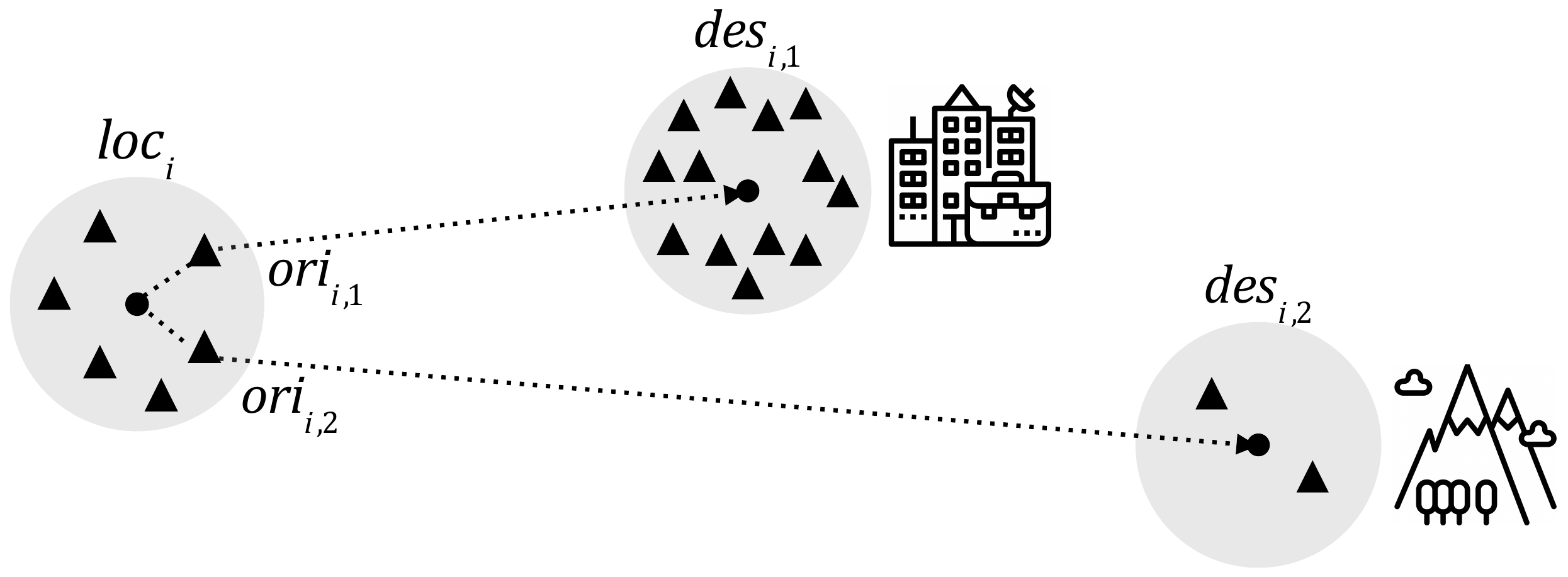}
		\caption{}
		\label{subfig:orderdispatching}
	\end{subfigure}
	\begin{subfigure}[b]{.33\linewidth}
		\centering
		\includegraphics[height=1.2in]{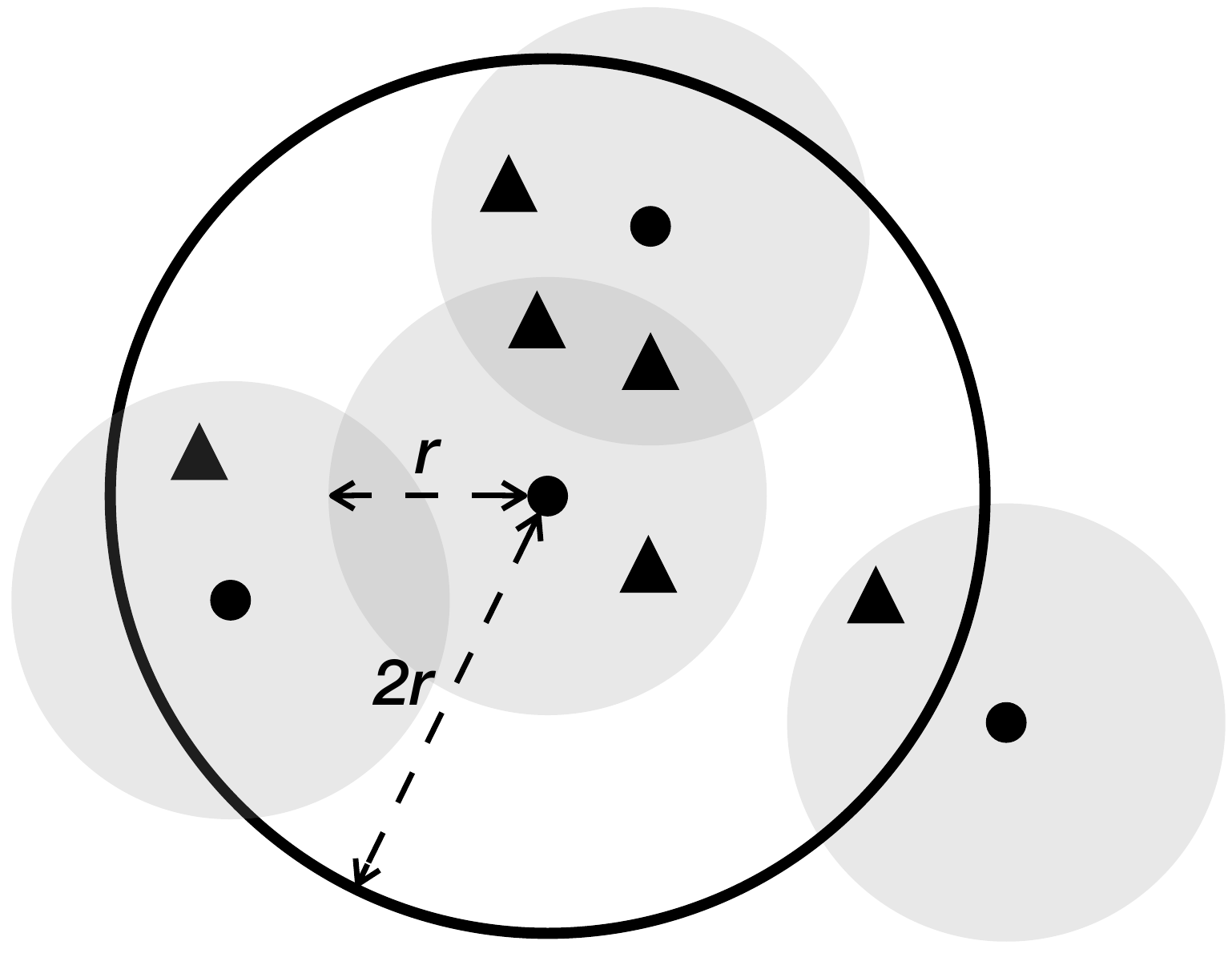}
		\caption{}
		\label{subfig:neighborhood}
	\end{subfigure}
  \caption{The order dispatching problem. (a) Two order choices (black triangle, departing from $ori_{i,1}$ and $ori_{i,2}$ respectively) within a driver $i$'s (black dot, located at $loc_i$) order receiving area (grey shadow), where one ends at a neighboring CBD (located at $des_{i,1}$) with high demand, and the other ends at a distant suburb (located at $des_{i,2}$) with low demand. Both orders have the same pick-up distance for driver $i$. (b) Details within a neighborhood, where the radius of the order receiving area and the neighborhood (black circle) are $r$ and $2r$ respectively.}
  \label{fig:MAOD}
\end{figure*}

One direction to tackle the order dispatching challenge has been to apply hand-crafted features to either centralized dispatching authorities (e.g., the combinatorial optimization algorithm \cite{Papadimitriou:1982:COA:31027, Zhang:2017:TOD:3097983.3098138}) or distributed multi-agent scheduling systems \cite{Wooldridge:2009:IMS:1695886, alshamsi2009multiagent}, in which a group of autonomous agents that share a common environment interact with each other.
However, the system performance relies highly on the specially designed weighting scheme.
For centralized approaches, another critical issue is the potential "single point of failure" \cite{lynch2009single}, i.e., the failure of the centralized authority control will fail the whole system.
Although the multi-agent formulation provides a distributed perspective by allowing each driver to choose their order preference independently,
existing solutions require rounds of direct communications between agents during execution \cite{seow2010collaborative}, thus being limited to a local area with a small number of agents. 

Recent attempts have been made to formulate this problem with centralized authority control and model-free reinforcement learning (RL) \cite{sutton1998reinforcement}, which learns a policy by interacting with a complex environment.
However, existing approaches \cite{Xu:2018:LOD:3219819.3219824} formulate the order dispatching problem with the single-agent setting, which is unable to model the complex interactions between drivers and orders, thus being oversimplifications of the stochastic demand-supply dynamics in large-scale ridesharing scenarios.
Also, executing an order dispatching system in a centralized manner still suffers from the reliability issue mentioned above that is generally inherent to the centralized architecture.

Staying different from these approaches, in this work we model the order dispatching problem with multi-agent reinforcement learning (MARL) \cite{Busoniu2010}, where agents share a centralized judge (the critic) to rate their decisions (actions) and update their strategies (policies).
The centralized critic is no longer needed during the execution period as agents 
can follow their learned policies independently, making the order dispatching system more robust to potential hardware or connectivity failures.
With the recent development of the Internet of Things (IoT) \cite{6740844, GUBBI20131645, ATZORI20102787} and the Internet of Vehicles (IoV) \cite{6823640, 6969789},
the fully distributed execution could be practically deployed by distributing the centralized trained policy to each vehicle through the Vehicle-to-Network (V2N) \cite{abboud2016interworking}.
By allowing each driver to learn by maximizing its cumulative reward through time, the reinforcement learning based approach relieves us from designing a sophisticated weighting scheme for the matching algorithms.
Also, the multi-agent setting follows the distributed nature of the peer-to-peer ridesharing problem, providing the dispatching system with the ability to capture the stochastic demand-supply dynamics in large-scale ridesharing scenarios.
Meanwhile, such fully distributed executions also enable us to scale up to much larger scenarios with many more agents, i.e., a scalable real-time order dispatching system for ridesharing services with millions of drivers.

Nonetheless, the major challenge in applying MARL to order dispatching lies in the changing dynamics of two components: the size of the action set and the size of the population.
As illustrated in Fig. \ref{subfig:neighborhood}, the action set for each agent is defined as a set of neighboring active orders within a given radius $r$ from each active driver (shown as the shadow area around each black dot).
As the active orders will be taken and new orders keep arriving, the size and content of this set will constantly change over time.
The action set will also change when the agent moves to another location and arrives in a new neighborhood. On the other hand, 
drivers can also switch between online and offline in the real-world scenario,
the population size for the order dispatching task is therefore also changing over time.

In this paper, we address the two problems of variable action sets and population size by extending the actor-critic policy gradient methods.
Our methods tackle the order dispatching problem within the framework of centralized training with decentralized execution.
The critic is provided with the information from other agents to incorporate the peer information, while each actor behaves independently with local information only.
To resolve the variable population size, we adopt the mean field approximation to transform the interactions between agents to the pairwise interaction between an agent and the average response from a sub-population in the neighborhood.
We provide the convergence proof of mean field reinforcement learning algorithms with function approximations to justify our algorithm in theory.
To solve the issue of changing action set, we use the vectorized features of each order as the network input to generate a set of ranking values, which are fed into a Boltzmann softmax selector to choose an action.
Experiments on the large-scale simulator show that compared with variant multi-agent learning benchmarks, the mean field multi-agent reinforcement learning algorithm gives the best performance towards the GMV, the order response rate, and the average pick-up distance. 
Besides the state of the art performance, our solution also enjoys the advantage of distributed execution, which has lower latency and is easily adaptable to be deployed in the real world application.

\section{Method}
In this section, we first illustrate our definition of order dispatching as a Markov game, and discuss two challenges when applying MARL to this game.
We then propose a MARL approach with the independent $Q$-learning, namely the independent order dispatching algorithm (IOD), to solve this game.
By extending IOD with mean field approximations, which capture dynamic demand-supply variations by propagating many local interactions between agents and the environment, we finally propose the cooperative order dispatching algorithm (COD).

\subsection{Order Dispatching as a Markov Game}
\subsubsection{Game Settings}
\label{subsec:formation}
We model the order dispatching task by a Partially Observable Markov Decision Process (POMDP)~\cite{Littman:1994:MGF:3091574.3091594} in a fully cooperative setting, defined by a tuple $\Gamma={\langle}\St, \mathcal{P}, \Ac, \mathcal{R}, \mathcal{O}, N, \gamma {\rangle}$, where $\St, \mathcal{P}, \Ac, \mathcal{R}, \mathcal{O}, N, \gamma$ are the sets of states, transition probability functions, sets of joint actions, reward functions, sets of private observations, number of agents, and a discount factor respectively.
Given two sets $\mathcal{X}$ and $\mathcal{Y}$, we use $\mathcal{X} \times \mathcal{Y}$ to denote the Cartesian product of $\mathcal{X}$ and $\mathcal{Y}$, i.e.,
$\mathcal{X} \times \mathcal{Y} = \{(x, y) | x \in \mathcal{X}, y \in \mathcal{Y}\}$.
The definitions are given as follows:

\leftmargini=4mm
\begin{itemize}
\item $N$: $N$ homogeneous agents identified by $i \in \mathcal{I} \equiv \{1,...,N\}$ are defined as the active drivers
in the environment.
As the drivers can switch between online and offline via a random process,
the number of agents $N$ could change over time.

\item $\St, \mathcal{O}$: At each time step $t$, agent $i$ draws private observations $o_i^t \in \mathcal{O}$ correlated with the true environment state $s^t \in \St$ according to the observation function $\mathcal{S} \times \mathcal{I} \rightarrow \mathcal{O}$.
The initial state of the environment is determined by a distribution $p_1(s^1): \St \rightarrow [0,1]$.
A typical environmental state $s$ includes the order and driver distribution, the global timestamp, and other environment dynamics (traffic congestion, weather conditions, etc.).
In this work, we define the observation for each agent $i$ with three components: agent $i$'s location $loc_i$, the timestamp $t$, and an on-trip flag to show if agent $i$ is available to take new orders. 

\item $\mathcal{P}, \Ac$: In the order dispatching task, each driver can only take active orders within its neighborhood (inside a given radius, as illustrated in Fig. \ref{subfig:neighborhood}). Hence, agent $i$'s action set $\Ac_i$ is defined as its own active order pool based on the observation $o_i$.
Each action candidate $a_{i,m} \in \Ac_i$ is parameterized by the normalized vector representation of corresponding order's origin $ori_{i,m}$ and destination $des_{i,m}$, i.e., $a_{i,m} \equiv [ori_{i,m}, des_{i,m}]$.
At time step $t$, each agent takes an action $a_i^t \in \mathcal{A}_i$,
forming a set of joint driver-order pair $\mathbf{a}^t = \Ac_1 \times ... \times \Ac_N$, which induces a transition in the environment according to the state transition function
\begin{equation}
\mathcal{P}(s^{t+1}|s^t,\mathbf{a}^t): \St \times \Ac_1 \times ... \times \Ac_N \rightarrow\St.
\label{state_transition_function}
\end{equation}
For simplicity, we assume no order cancellations and changes during the trip, i.e., $a_i^t$ keeps unchanged if agent $i$ is on its way to the destination.

\item $\expR$:
Each agent $i$ obtains rewards $r_i^t$ by a reward function
\begin{equation}
\expR_i(s^t,\mathbf{a}^t): \St \times \Ac_1 \times ... \times \Ac_N \rightarrow \mathbb{R}.
\label{reward_function}
\end{equation}
As described in Section \ref{sec:intro}, we want to maximize the total income by considering both the charge of each order and the potential opportunity of the destination.
The reward is then defined as the combination of driver $i$'s own income $\prescript{0}{}{r}$ from its order choice and the order destination potential $\prescript{1}{}{r}$
, which is determined by all agents' behaviors in the environment.
Considering the credit-assignment problem \cite{Agogino08analyzingand} arises in MARL with many agents, i.e., the contribution of an agent's behavior is drowned by the noise of all the other agents' impact on the reward function, we set each driver's own income instead of the total income of all drivers as the reward.

\begin{figure*}
  \vskip -0.1in
  \centering
	\begin{subfigure}[b]{.47\linewidth}
		\centering
		\includegraphics[height=0.42\columnwidth]{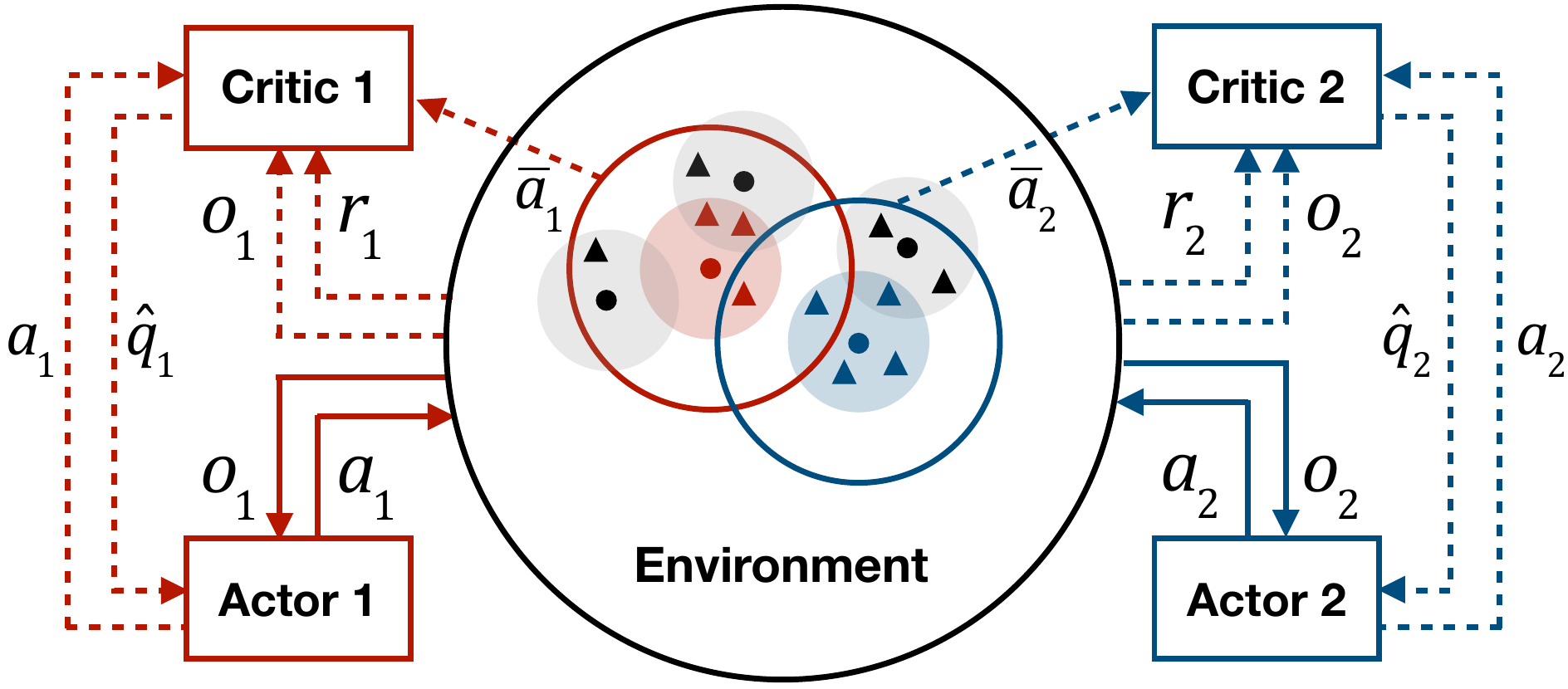}
		\caption{}
		\label{subfig:arc}
	\end{subfigure}
	\begin{subfigure}[b]{.26\linewidth}
		\centering
		\includegraphics[height=0.8\columnwidth]{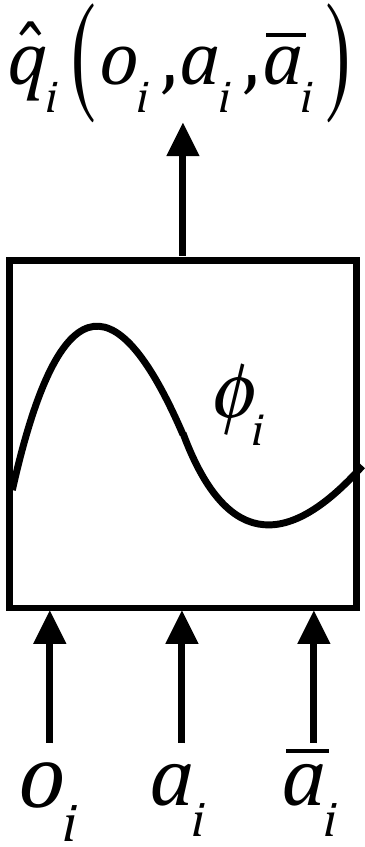}
		\caption{}
		\label{subfig:critic}
	\end{subfigure}
	\begin{subfigure}[b]{.25\linewidth}
		\centering
		\includegraphics[height=0.8\columnwidth]{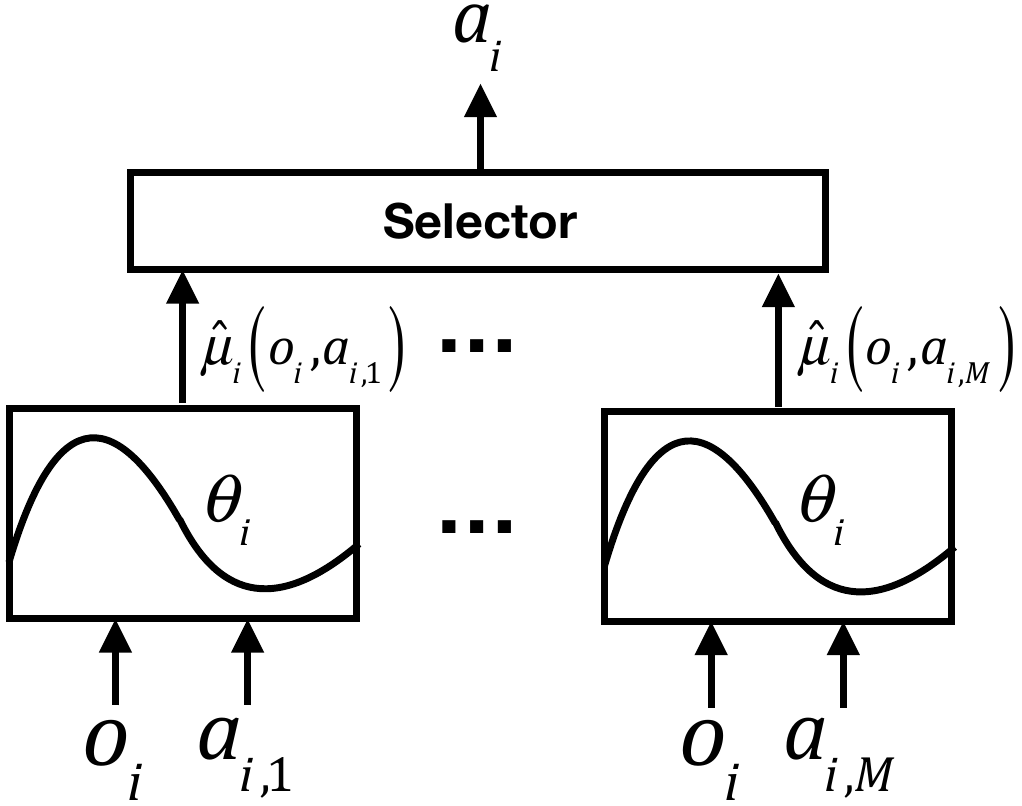}
		\caption{}
		\label{subfig:actor}
	\end{subfigure}
  \caption{Overview of our COD approach.
  (a) The information flow between RL agents and the environment (shown with two agents).
  The centralized authority is only needed during the training stage (dashed arrows) to gather the average response $\bar{a}$;
  agents could behave independently during the execution stage (solid arrows), thus being robust to the "single point of failure".
  (b) and (c) Architectures of the critic and the actor.}
  \label{fig:MAOD}
\end{figure*}

To encourage cooperation between agents and avoid agents' being selfish and greedy, we use the order destination's demand-supply gap as a constraint $\prescript{1}{}{r}$ on the behavior of agents.
Precisely, we compare the demand-supply status between the order origin and destination, and encourage the driver to choose the order destination with a larger demand-supply gap.
The order destination potential (DP) is defined as
\begin{equation}
DP = \#DD - \#DS,
\label{destination_potential}
\end{equation}
where \#DD and \#DS is the demand and the supply of the destination respectively.
We consider the DP only if the number of orders are larger than that of drivers at the origin.
If the order destination has more drivers than orders, we penalize this order with the demand-supply gap at the destination, and vice versa.

To provide better customer experience, we also add the pick-up distance $\prescript{2}{}{r}$ as a regularizer. The ratio of the DP and the pick-up distance to order price are defined as $\prescript{1}{}{\alpha}$ and $\prescript{2}{}{\alpha}$ respectively, i.e., $r_i^t = \prescript{0}{}{r}_i^t + \prescript{1}{}{\alpha}\prescript{1}{}{r}_i^t + \prescript{2}{}{\alpha}\prescript{2}{}{r}_i^t$.
We typically choose the regularization ratio to scale different reward terms into approximately the same range,
although in practice a grid-search could be used to get better performance.
The effectiveness of our reward function setting is empirically verified in Section \ref{subsubsec:grid_result}.

\item $\gamma$: Each agent $i$ aims to maximize its total discounted reward
\begin{equation}
G^t_i = \sum\nolimits_{k=t}^\infty\gamma^{k-t}r_i^k
\label{total_discounted_reward}
\end{equation}
from time step $t$ onwards, where $\gamma \in [0,1]$ is the discount factor.
\end{itemize}

We denote joint quantities over agents in bold, and joint quantities over agents other than a given agent $i$ with the subscript ${-i}$, e.g.,\ $\vect{a} \equiv (\vect{a}_{-i}^t,a_{i}^t)$. To stabilize the training process, we maintain an experience replay buffer $\mathcal{D}$ containing tuples $(\vect{o}^t, \vect{a}^t, \vect{r}^t, \vect{o}^{t+1})$ as described in ~\cite{mnih2015human}.

\subsubsection{Dynamic of the Action Set Elements}
In the order dispatching problem, an action is defined as an active order within a given radius from the agent.
Hence, the content and size of the active order pool for each driver are changing with both the location and time, i.e., the action set for an agent in the order dispatching MDP is changing throughout the training and execution process.
This aspect of order dispatching refrains us from using the $Q$-table to log the $Q$-value because of the potentially infinite size of the action set.
On the other hand, a typical policy network for stochastic actions makes use of a softmax output layer to produce a set of probabilities 
of choosing each action among a fixed set of action candidates, thus is unable to fit into the order dispatching problem.

\subsubsection{Dynamic of the Population Size}
To overcome the non-stationarity of the multi-agent environment, ~\citet{lowe2017multi} uses $Q$-Learning to approximate the discounted reward $G^t_i$, and rewrites the gradient of the expected return for agent $i$ following a deterministic policy $\mu_i$ (parameterized by $\theta_i$) as
\begin{equation}\label{eqn:pg_MAAC}
\nabla_{\theta_i} J(\theta_i) = \Expect_{\vect{o},\vect{a}\sim \mathcal{D}} [\nabla_{\theta_i} \mu_i(a_i|o_i) \nabla_{a_i} Q_i (o_i,\vect{a})]. 
\end{equation}
Here $Q_i (o_i,\vect{a})$ is a centralized action-value function that takes as input the observation $o_i$ of agent $i$ and the joint actions $\vect{a}$ of all agents, and outputs the $Q$-value for agent $i$. However, as off-line
drivers cannot participate in the order dispatching procedure, the number of agents $N$ in the environment is changing over time.
Also, in a typical order dispatching task, which involves thousands of agents, the high dynamics of interactions between a large number of agents is intractable.
Thus, a naive concatenation of all other agents' actions cannot form a valid input for the value network and is not applicable to the order dispatching task.

\subsection{Independent Order Dispatching}
\label{subsec:iod}
To solve the order dispatching MDP, we first propose the independent order dispatching algorithm (IOD), a straightforward MARL approach with the independent $Q$-learning.
We provide each learner with the actor-critic model, which is a popular form of policy gradient (PG) method.
For each agent $i$, PG works by directly adjusting the parameters $\theta_i$ of the policy $\mu_i$ to maximize the objective
$J{(\theta_i)} = \Expect_{\vect{o},\vect{a}\sim \mathcal{D}} \left[G^t_i\right]$
by taking steps in the direction of
$\nabla_{\theta_i} J(\theta_i)$. 
In MARL, independent actor-critic uses the action-value function $Q_i (o_i,a_i)$ to approximate the discounted reward $G^t_i$ by $Q$-Learning \cite{watkins1992q}. 
Here we use temporal-difference learning \cite{sutton1988learning} to approximate the true $Q_i (o_i,a_i)$, leading to a variety of actor-critic algorithms with $Q_i (o_i,a_i)$ called as the critic and $\mu_i$ called as the actor.

To solve the problem of variable action sets, we use a policy network with both observation and action embeddings as the input, derived from the in-action approximation methods (shown in Fig. \ref{subfig:critic}).
As illustrated in Fig. \ref{subfig:actor},
we use a deterministic policy $\mu_{{\ap}_i}$ (denoted by $\mu_{{\ap}_i}: \mathcal{O}\times\Ac\rightarrow \mathbb{R}^*$, abbreviated as $\mu_i$) to generate ranking values of each observation-action pair $(o_i,a_{i,m})$ for each of the $M_i$ candidates within agent $i$'s action set $\Ac_i$.
To choose an action $a_i$, these values are then fed into a Boltzmann softmax selector
\begin{equation}
\pi_{i}(a_{i,j} | o_{i}) = \frac{\exp{\beta \mu_{i}(o_i, a_{i,j})}}
{\sum_{m=0}^{M_i}{\exp{\beta \mu_{i}(o_i, a_{i,m})}}}~, \quad \text{for } j = 1,\dots,M_i
\label{boltzman_selector}
\end{equation}
where $\beta$ is the temperature to control the exploration rate.
Note that a typical policy network with out-action approximation is equivalent to this approach, where we can use an $M_i$-dimension one-hot vector as the embedding to feed into the policy network.
The main difference between these approaches is the execution efficiency, as we need exactly $M_i$ forward passes in a single execution step.
Meanwhile, using order features naturally provide us with an informative form of embeddings.
As each order is parameterized by the concatenation of the normalized vector representation of its origin $ori_{i,m}$ and destination $des_{i,m}$, i.e., $a_{i,m} \equiv [ori_{i,m}, des_{i,m}]$.
Similar orders will be close to each other in the vector space and produce similar outputs from the policy network, which improves the generalization ability of the algorithm.

In IOD, each critic takes input the observation embedding $o_i$ by combining
agent $i$'s location and the timestamp $t$.
The action embedding is built with the vector representation of the order destination and the distance between the driver location and the order origin.
The critic is a DQN \cite{mnih2015human} using neural network function approximations to learn the action-value function $Q_\phi$ (parameterized by $\phi$, abbreviated as $Q_i$ for each agent $i$) by minimizing the loss
\begin{equation} \label{eqn:critic}
\mathcal{L}(\phi_i) = \mathbb{E}_{\vect{o},\vect{a},\vect{r},\vect{o}'}[(y-Q_i(o_i, a_i))^2], \;\;
y=r_i+\gamma\, Q^-_i({o}_i',\mu_i^-(o_i', a_i')),
\end{equation}
where $Q^-_i$ is the target network for the action-value function, and $\mu^-_i$ is the target network for the deterministic policy.
These earlier snapshots of parameters are periodically updated with the most recent network weights and help increase learning stability by decorrelating predicted and target $Q$-values and deterministic policy values.

Following \citet{silver2014deterministic}, the gradient of the expected return for agent $i$ following a deterministic policy $\mu_i$ is
\begin{equation}\label{eqn:pg_MAOD}
\nabla_{\theta_i} J(\theta_i) = \Expect_{\vect{o},\vect{a}\sim \mathcal{D}} [\nabla_{\theta_i} \mu_i(o_i, a_i) \nabla_{a_i} Q_i (o_i,a_i)].
\end{equation}
Here $Q_i (o_i,a_i)$ is an action-value function that takes as input the observation $o_i$ and action $a_i$, and outputs the $Q$-value for agent $i$.

\begin{algorithm}
	\caption{Cooperative Order Dispatching (COD)}\label{algo:cod}
 \begin{algorithmic}
	\State Initialize $Q_{\phi_i}$, $Q_{\phi_i}^-$, $\mu_{\theta_i}$, and $\mu_{\theta_i}^-$ for all $i \in \{1,\dots,N\}$
	\While{training not finished}
	\State For each agent $i$, sample action $a_i$ using the Boltzmann softmax selector $\pi_i(a_i|o_i)$ from Eq.~\eqref{boltzman_selector}
	\State Take the joint action $\vs{a} = [a_1,\dots,a_N]$ and observe the reward $\vs{r} = [r_1,\dots,r_N]$ and the next observations $\vs{o}'$
	\State Compute the new mean action $\bar{\vs{a}} = [\bar{a}_1,\dots,\bar{a}_N]$
	\State Store $\langle \vs{o},\vs{a},\vs{r},\vs{o}',\bar{\vs{a}} \rangle$ in replay buffer $\mathcal{D}$
	\For{$i=1\mbox{~to~}N$}
	\State Sample $K$ experiences $\langle \vs{o},\vs{a},\vs{r},\vs{o}',\bar{\vs{a}} \rangle$ from $\mathcal{D}$
	\State Update the critic by minimizing the loss from Eq.~\eqref{eqn:mf-ac-critic}
	\State Update the actor using the policy gradient as Eq.~\eqref{eqn:mf-ac-actor}
	\EndFor     
	\State Update the parameters of the target networks for each agent $i$ with updating rates $\tau_{\phi}$ and $\tau_{\theta}$:
	\begin{align*}
		\phi_i^- &\gets \tau_{\phi}\phi_i+(1-\tau_{\phi})\phi_i^-\\
		\theta_i^- &\gets \tau_{\theta}\theta_i+(1-\tau_{\theta})\theta_i^-
	\end{align*}
	\EndWhile
\end{algorithmic}
\end{algorithm}

\subsection{Cooperative Order Dispatching with Mean Field Approximation}
\label{subsec:MFOD}
To fully condition on other agents' policy in the environment with variable population size, we propose to integrate our IOD algorithm with mean field approximations, following the Mean Field Reinforcement Learning (MFRL) \cite{pmlr-v80-yang18d}.
MFRL addresses the scalability issue in the multi-agent reinforcement learning with a large number of agents, where the interactions are approximated pairwise by the interaction between an agent and the average response $\bar{a}$ from a sub-population in the neighborhood.
As this pairwise approximation shadows the exact size of interacting counterparts, the use of mean field approximation can help us model other agents' policies directly in the environment with variable population sizes.

In the order dispatching task, agents are interacting with each other by choosing order destinations with a high demand to optimize the demand-supply gap. 
As illustrated in Fig. \ref{subfig:neighborhood}, the range of the neighborhood is then defined as twice the length of the order receiving radius, because agents within this area have intersections between their action sets and interact with each other.
The average response $\bar{a}_i$ is therefore defined as the number of drivers arriving at the same neighborhood as agent $i$, divided by the number of available orders for agent $i$.
For example, when agent $i$ finishes the order and arrives at the neighborhood in Fig. \ref{subfig:neighborhood} (the central agent), the average response $\bar{a}_i$ is 2/3, as there are two agents within the neighborhood area and three available orders.

The introduction of mean field approximations enables agents to learn with the awareness of interacting counterparts, thus helping to improve the training stability and robustness of agents after training for the order dispatching task.
Note that the average response $\bar{a}$ only serves for the model update; thus the centralized authority is only needed during the training stage.
During the execution stage, agents could behave in a fully distributed manner, thus being robust to the "single point of failure".

We propose the cooperative order dispatching algorithm (COD)
as illustrated in Fig. \ref{subfig:arc},
and present the pseudo code for COD in Algorithm \ref{algo:cod}. Each critic is trained by minimizing the loss
\begin{equation}\label{eqn:mf-ac-critic}
\mathcal{L}(\phi_i)=\mathbb{E}_{\vect{o},\vect{a},\vect{r},\vect{o}'}[(r_i+\gamma\, v^{\mf}_{i^-}(o_i') - Q_i(o_i, (\bar{a}_i, a_i)))^2],
\end{equation}
where $v^{\mf}_{i^-}(o_i')$ is the mean field value function for the target network $Q_i^-$ and $\mu_i^-$ (shown as the Boltzmann selector $\pi_i^-$ from Eq. \ref{boltzman_selector})
\begin{equation}
\label{mfv}
v^{\mf}_{i^-}(o_i') = \sum\nolimits_{a_i'} \pi_i^-(a_i' | o_i') \mathbb{E}_{\bar{a}_i'
(\vs{a}'_{-i})\sim\vs{\pi}_{-i}^-
} [ Q_i^- ( o_i', (\bar{a}'_i, a'_i) )],
\end{equation}
and $\bar{a}_i$ is the average response within agent $i$'s neighborhood. The actor of COD learns the optimal policy by using the policy gradient:
\begin{equation}\label{eqn:mf-ac-actor}
\nabla_{\theta_i} J(\theta_i) = \Expect_{\vect{o},\vect{a}\sim \mathcal{D}} [\nabla_{\theta_i} \mu_i(o_i, a_i) \nabla_{a_i} Q_i (o_i, (\bar{a}_i, a_i))]. 
\end{equation}

In the current decision process, active agents sharing duplicated order receiving areas (e.g., the central and upper agent in Fig. \ref{subfig:neighborhood}) might select the same order following their own strategy.
Such collisions could lead to invalid order assignment and force both drivers and customers to wait for a certain period, which equals to the time interval between each dispatching iteration.
Observe that the time interval of decision-making also influences the performance of dispatching; a too long interval will affect the passenger experience, while a too short interval without enough order candidates is not conducive to the decision making.
To solve this problem, our approach works in a fully distributed manner with asynchronous dispatching strategy,
allowing agents have different decision time interval for individual states,
i.e., agents assigned with invalid orders could immediately re-choose a new order from updated candidates pool.

To theoretically support the efficacy of our proposed COD algorithm, we provide the convergence proof of MFRL with function approximations as shown below.

\subsection{Convergence of Mean Field Reinforcement Learning with Function Approximations}

Inspired by the previous proof of MFRL convergence in a tabular $Q$-function setting \cite{pmlr-v80-yang18d}, we further develop the proof towards the converge when the $Q$-function is represented by other function approximators. In addition to the Markov Game setting in Section \ref{subsec:formation}, let $\mathcal{Q}=\{\vs{Q}_\phi \}$ be a family of real-valued functions defined on $\mathcal{S}\times \mathcal{A} \times \bar{\mathcal{A}}$, where $\bar{\mathcal{A}}$ is the action space for the mean actions computed from the neighbors. For simplicity, we assume the environment is a fully observable MDP $\Gamma$, i.e., each agent $i \in \mathcal{I} \equiv \{1,...,N\}$ can observe the global state $s$ instead of the local observation $o_i$.

Assuming that the function class is linearly parameterized, for each agent $i$, the $Q$-function can be expressed as the linear span of a fixed set of $P$ linearly independent functions $\omega^p_i: \mathcal{S}\times \mathcal{A} \times \bar{\mathcal{A}} \rightarrow \mathbb{R}$. Given the parameter vector $\phi_i \in \mathbb{R}^P$, the function $Q_{\phi_i}$ (abbreviated as $Q_i$) is thus defined as
\begin{equation}
Q_i(s, (\bar{a}_i, a_i))=\sum_{p=1}^P \omega_i^p(s, (\bar{a}_i, a_i)) \phi_i(p) = \omega_i(s, (\bar{a}_i, a_i))^\top \phi_i.
\end{equation}
In the function approximation setting, we apply the update rules:
\begin{equation}
\phi_i'
= \phi_i + \alpha \Delta   \nabla_{\phi_i}{Q_i(s, (\bar{a}_i, a_i))}
= \phi_i + \alpha \Delta \omega_i(s, (\bar{a}_i, a_i)) \label{itersteps},
\end{equation}
where $\Delta$ is the temporal difference:
\begin{align}
\Delta
\nonumber &=  r_i+\gamma\, v^{\mf}_i(s') - {Q_i(s, (\bar{a}_i, a_i))} \\
 &= r_i+\gamma\, \mathbb{E}_{\vs{a} \sim \vs{\pi}^{\mf}} [ Q_i(s', (\bar{a}_i, a_i)) ] - Q_i(s, (\bar{a}_i, a_i)).
\end{align}
Our goal is to derive the parameter vector $\vs\phi = \{\phi_i \}$ such that ${\vs\omega}^{\top} \vs\phi$  approximates the (local) Nash $Q$-values.
Under the main assumptions and the lemma as introduced below,
\citet{pmlr-v80-yang18d} proved that the policy $\vs{\pi}$ is $\frac{K}{2T}$ Lipschitz continuous with respect to $\phi$, where $T=1/\beta$ and  $K \geq 0$ is the upper bound of the observed reward. 

\begin{assumption}\label{tableassum}
Each state-action pair is visited infinitely often, and the reward is bounded by some constant $K$.
\end{assumption}

\begin{assumption}\label{glieassum}
Agent's policy is Greedy in the Limit with Infinite Exploration (GLIE). In the case with the Boltzmann policy, the policy becomes greedy \emph{w.r.t.} the $Q$-function in the limit as the temperature decays asymptotically to zero. 
\end{assumption}

\begin{assumption}\label{2nasheq}
For each stage game $[Q_t^1(s), ..., Q_t^N(s)]$ at time $t$ and in state $s$ in training, for all $t$, $s$, $j\in \{1,\dots,N\}$, 
the Nash equilibrium $\vs{\pi}_* = [\pi^1_*, \dots, \pi^N_*]$ is recognized either as 1) the \emph{global optimum} or 2) a \emph{saddle point} expressed as:
\begin{itemize}
\item[1.] $\mathbb{E}_{\vs{\pi}_*} [Q_t^j(s)] \ge \mathbb{E}_{\vs{\pi}} [Q_t^j(s)],\ \forall\vs{\pi} \in \Omega\big(\prod_k \mathcal{A}^k\big)$;
\item[2.] $\mathbb{E}_{\vs{\pi}_*} [Q_t^j(s)] \ge \mathbb{E}_{\pi^j}\mathbb{E}_{\vs{\pi}^{-j}_*} [Q_t^j(s)],\ \forall\pi^j \in \Omega\big(\mathcal{A}^j\big)$ and\\$\mathbb{E}_{\vs{\pi}_*} [Q_t^j(s)] \le \mathbb{E}_{\pi^j_*}\mathbb{E}_{\vs{\pi}^{-j}} [Q_t^j(s)],\ \forall\vs{\pi}^{-j} \in \Omega\big(\prod_{k\ne j}\mathcal{A}^k\big)$.
\end{itemize}
\end{assumption}

\begin{lemma}\label{fundamental}
The random process $\{\Delta_t\}$ defined in $\mathbb{R}$ as
\begin{equation*}
\Delta_{t+1}(x) = (1 - \alpha_t(x))\Delta_{t}(x) + \alpha_t(x)F_t(x)
\label{lemma1}
\end{equation*}
converges to zero with probability $1$ (\emph{w.p.$1$}) when
\begin{itemize}
\item[1.] $0\leq \alpha_t(x) \leq 1$, $\sum_t \alpha_t(x) = \infty$, $\sum_t \alpha_t^2(x) < \infty$;
\item[2.] $x \in \mathcal{X}$, the set of possible states, and $|\mathcal{X}| < \infty$;
\item[3.] $\| \mathbb{E}[F_t(x) | \mathcal{F}_t ] \|_{W} \leq \gamma \|\Delta_t\|_W + c_t$, where $\gamma \in [0,1)$ and $c_t$ converges to zero \emph{w.p.$1$};
\item[4.] $\mathbf{var}[F_t(x) | \mathcal{F}_t ] \leq K(1+\|\Delta_t\|_{W}^2)$ with constant $K>0$.
\end{itemize}
Here $\mathcal{F}_t$ denotes the filtration of an increasing sequence of $\sigma$-fields including the history of processes; $\alpha_t, \Delta_t, F_t \in \mathcal{F}_t$ and $\|\cdot\|_W$ is a weighted maximum norm \cite{bertsekas2012weighted}.
\end{lemma}
\begin{proof}
See Theorem 1 in \citet{jaakkola1994convergence} and Corollary 5 in \citet{szepesvari1999unified} for detailed derivation.
We include it here to stay self-contained.
\end{proof}

In contrast to the previous work \citet{pmlr-v80-yang18d}, we establish convergence of Eq.~\eqref{itersteps} by adopting an ordinary differentiable equation (ODE) with a globally asymptotically stable equilibrium point where the trajectories closely follow, following the framework of the convergence proof of single-agent $Q$-learning with function approximation \cite{melo2008analysis}.

\begin{theorem}
Given the MDP $\Gamma$, $\{\vs\omega_p, p =1, ..., P \}$, and the learning policy  $\vs{\pi}$ that is $\frac{K}{2T}$ Lipschitz continuous with respect to $\phi$, if the Assumptions \ref{tableassum}, \ref{glieassum} \& \ref{2nasheq}, and Lemma \ref{fundamental}'s first and second conditions
are met, then there exists $C_0$ such that the algorithm in Eq.~\eqref{itersteps} converges w.p.1 if $\frac{K}{2T} < C_0$.
\end{theorem}

\begin{proof}
We first re-write the Eq.~\eqref{itersteps} as on ODE:
\begin{align}
\dfrac{d\vs\phi}{dt}& = \mathbb{E}_{\vs\phi}\left[\vs\omega_s^{\top}\left(\vs{r}(s, \vs\mathbf{a}, s') + \gamma \vs\omega_{s'}^{\top}\vs\phi  - \vs\omega_{s}^{\top}\vs\phi  \right)  \right]	\nonumber \\
& =  \mathbb{E}_{\vs\phi}\left[\vs\omega_s^{\top}(\gamma \vs\omega_{s'}^{\top} - \vs\omega_{s}^{\top} )\right] \vs\phi + \mathbb{E}_{\vs\phi}\left[\vs\omega_s^{\top}(\vs{r}(s, \vs\mathbf{a}, s'))\right] \nonumber \\
& = \vs{A}_{\phi} \vs\phi + \vs{b}_{\phi} 
\label{ode}
\end{align}
Notice that we use a vector for considering the updating rule for the $Q$-function of each agent.
We can easily know that necessity condition of the equilibrium is
$\vs\phi^* = \vs{A}_{\phi^*}^{-1}\vs{b}_{\phi^*}$.
The existence of the such equilibrium has been restricted in the scenario that meets Assumption \ref{2nasheq}.
\citet{pmlr-v80-yang18d} proved that under the Assumption \ref{2nasheq}, the existing equilibrium, either in the form of a global equilibrium or in the form of a saddle-point equilibrium, is unique.

Let $\tilde{\vs\phi} = \vs\phi^t - \vs\phi^*$, we have:
\begin{align}
\dfrac{d}{dt} || \tilde{\vs\phi}||_{2}  & = 2 \vs\phi \cdot \frac{d\vs\phi}{dt} - 2 \vs\phi^*\cdot \frac{d\vs\phi}{dt}  \nonumber \\
&= 2 \tilde{\vs\phi}^{\top}(\vs{A}_{\phi} \vs\phi + \vs{b}_{\phi} -\vs{A}_{\phi^*} \vs\phi^* - \vs{b}_{\phi^*}) \nonumber \\
& = 2 \tilde{\vs\phi}^{\top}(\vs{A}_{\phi^*} \vs\phi - \vs{A}_{\phi^*} \vs\phi + \vs{A}_{\phi} \vs\phi + \vs{b}_{\phi} -\vs{A}_{\phi^*} \vs\phi^* - \vs{b}_{\phi^*})  \nonumber \\
& = 2 \tilde{\vs\phi}^{\top}\vs{A}_{\phi^*} \tilde{\vs\phi} + 2 \tilde{\vs\phi}^{\top}(\vs{A}_{\phi} - \vs{A}_{\phi^*})\vs\phi + 2 \tilde{\vs\phi}^{\top}(\vs{b}_{\phi} - \vs{b}_{\phi^*}) \nonumber \\
& \leq  2 \tilde{\vs\phi}^{\top} \left(\vs{A}_{\phi^*} + \sup_{\vs\phi} || \vs{A}_{\phi} - \vs{A}_{\phi^*}||_2 + \sup_{\vs\phi} \dfrac{||\vs{b}_{\phi} - \vs{b}_\phi^*||_2}{||\vs\phi - \vs\phi^*||_2}   \right) \tilde{\vs\phi}
\label{odefinal}.
\end{align}
As we know that the policy $\vs\pi^t$ is Lipschitz continuous w.r.t $\phi^{t}$, this implies that $\vs{A}_{\phi}$ and $\vs{b}_{\phi}$ are also Lipschitz continuous w.r.t to $\phi$. In other words, if $\frac{K}{2T} \leq C_0$  is sufficiently small and close to zero, 
then the norm term of $\left(\sup_{\vs\phi} || \vs{A}_{\phi} - \vs{A}_{\phi^*}||_2 + \sup_{\vs\phi} \dfrac{||\vs{b}_{\phi} - \vs{b}_\phi^*||_2}{||\vs\phi - \vs\phi^*||_2}\right)$ goes to zero. Considering near the equilibrium point $\phi^*$, $\vs{A}_{\phi^*}$ is a negative definite matrix, the Eq.~\eqref{odefinal} tends to be negative definite as well, so the ODE in Eq.\eqref{ode} is globally asymptotically stable and the conclusion of the theorem follows.
\end{proof}

While in practice, we might break the linear condition by the use of nonlinear activation functions, the Lipschitz continuity will still hold as long as the nonlinear add-on is limited to a small scale.

\section{Experiment}
To support the training and evaluation of our MARL algorithm, we adopt two simulators with a grid-based map and a coordinate-based map respectively.
The main difference between these two simulators is the design of the driver pick-up module, i.e., the process of a driver reaching the origin of the assigned order.
In the grid-based simulator (introduced by \citet{lin2018efficient}), the location state for each driver and order is represented by a grid ID.
Hence, the exact coordinate for each instance inside a grid is shadowed by the state representation.
This simplified setting ensures there will be no pick-up distance (or arriving time) difference inside a grid; it also brings an assumption of no cancellations before driver pick-up.
Whereas in the coordinate-based simulator, the location of each driver and order instance is represented by a two-value vector from the Geographic Coordinate System, and the cancellation before pick-up is also taken into account.
In this setting, taking an order within an appropriate pick-up distance is crucial to each driver, as the order may be canceled if the driver takes too long time to arrive.
We present experiment details of the grid-based simulator in Section \ref{subsec:grid-based}, and the coordinate-based simulator in Section \ref{subsec:coord-based}.

\begin{figure}[t]
  \centering
  \includegraphics[height=0.35\columnwidth]{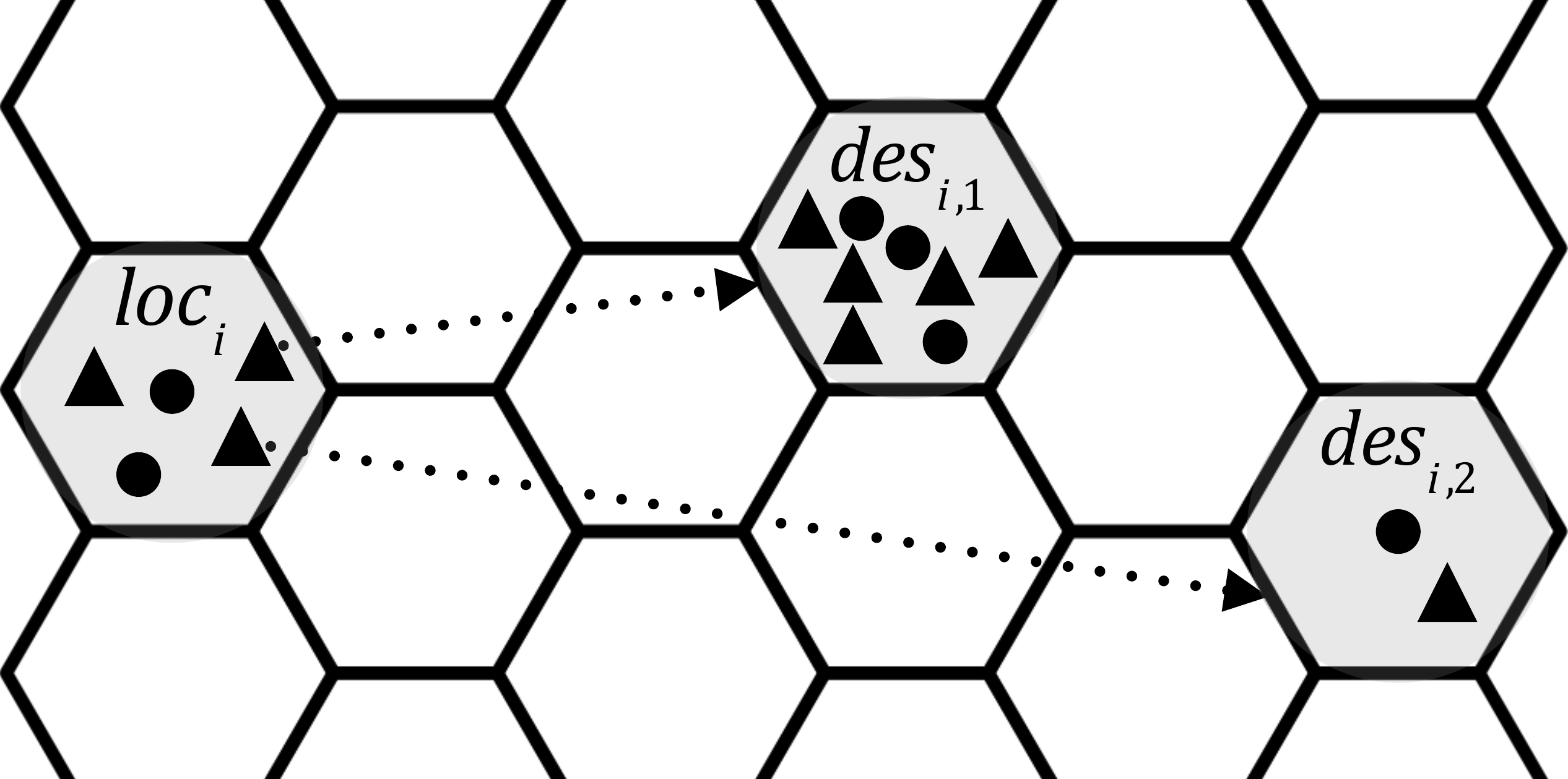}
  \caption{Illustration of the grid-based simulator. All location features are represented by the corresponding grid ID, where the order origin $ori_{i,m}$ is the same as the driver location $loc_i$.}
  \label{fig:grid_simulator}
\end{figure}

\subsection{Grid-based Experiment}
\label{subsec:grid-based}
\subsubsection{Environment Setting}
\label{subsubsec:grid-based-setting}
In the grid-based simulator, the city is covered by a hexagonal grid-world as illustrated in Fig. \ref{fig:grid_simulator}.
At each simulation time step $t$, the simulator provides an observation $\vect{o}^t$ with a set of active drivers and a set of available orders.
Each order feature includes the origin grid ID and the destination grid ID, while each driver has the grid ID as the location feature $loc_i$.
Drivers are regarded as homogeneous and can switch between online (active) and offline via a random process learned from the history data.
As the travel distance between neighboring grids is approximately $1.2$ kilometers and the time step interval $t_\triangle$ is $10$ minutes,
we assume that drivers will not move to other grids before taking a new order, and define the order receiving area and the neighborhood as the grid where the agent stays.
The order dispatching algorithm then generates an optimal list of driver-order pairs $\vect{a}^t$ for the current policy,
where $a_i$ is an available order $a_{i,m}$ selected from the order candidate pool $\Ac_i$.
In the grid-based setting, the origin of each order is already embedded as the location feature $loc_i$ in $o_i$, thus $a_{i,m}$ is parameterized by the destination grid ID $des_{i,m}$.
After receiving the driver-order pairs from the algorithm, the simulator will then return a new observation $\vect{o}^{t+1}$ and a list of order fees.
Stepping on this new observation, the order dispatching algorithm will calculate a set of rewards $r^t$ for each agent, store the record $(\vect{o}^t, \vect{a}^t, \vect{r}^t, \vect{o}^{t+1})$ to replay buffer, and update the network parameters with respect to a batch of samples from replay buffer.

The data source of this simulator (provided by DiDi Chuxing) includes order information and trajectories of vehicles in three weeks.
Available orders are generated by bootstrapping from real orders occurred in the same period during the day given a bootstrapping ratio $p_{sample}$.
More concretely, suppose the simulation time step interval is $t_\triangle$, at each simulation time step $t$, we randomly sample $p_{sample}$ orders with replacement from real orders happened between $t_\triangle * t$ to $t_\triangle * (t + 1)$.
Also, drivers are set between online and offline following a distribution learned from real data using a maximum likelihood estimation.
On average, the simulator has $7K$ drivers and $136K$ dispatching events per time step.

The effectiveness of the grid-based simulator is evaluated by \citet{lin2018efficient} using the calibration against the real data regarding the most important performance measurement: the gross merchandise volume (GMV).
The coefficient of determination $r^2$ between simulated GMV and real GMV is $0.9331$ and the Pearson correlation is $0.9853$ with $p$-value $p < 0.00001$.

\subsubsection{Model Setting}
\label{grid-based-setting}
We use the grid-based simulator to compare the performance of following methods.

\leftmargini=4mm
\begin{itemize}
\item \textbf{Random (RAN)}: The random dispatching algorithm considers no additional information. It only assigns all active drivers with an available order at each time step.

\item \textbf{Response-based (RES)}: This response-based method aims to achieve higher order response rate by assigning drivers to short duration orders.
During each time step, all available orders starting from the same grid will be sorted by the estimated trip time.
Multiple orders with the same expected duration will be further sorted by the order price to balance the performance.

\item \textbf{Revenue-based (REV)}: The revenue-based algorithm focuses on a higher GMV.
Orders with higher prices will be given priority to get dispatched first.
Following the similar principle as described above, orders with shorter estimated trip time will be assigned first if multiple orders have the same price.

\item \textbf{IOD}: The independent order dispatching algorithm as described in Section \ref{subsec:iod}.
The action-value function approximation (i.e., the $Q$-network) is parameterized by an MLP with four hidden layers (512, 256, 128, 64) and the policy network is parameterized by an MLP with three hidden layers (256, 128, 64).
We use the ReLU \cite{nair2010rectified} activation between hidden layers, and transform the final linear output of $Q$-network and policy network with ReLU and sigmoid function respectively.
To find an optimal parameter setting, we use the Adam Optimizer \cite{kingma2014adam} with a learning rate of $0.0001$ for the critic and $0.001$ for the actor.
The discounted factor $\gamma$ is $0.95$, and the batch size is $2048$.
We update the network parameters after every $3K$ samples are added to the replay buffer (capacity $5*10^5$).
We use a Boltzmann softmax selector for all MARL methods and set the initial temperature as $1.0$, then gradually reduce the temperature until $0.01$ to limit exploration.

\item \textbf{COD}: Our proposed cooperative order dispatching algorithm with mean field approximation as described in Section \ref{subsec:MFOD}.
The network architecture is identical to the one described in IOD, except a mean action is fed as another input to the critic network as illustrated in Fig. \ref{subfig:critic}.
\end{itemize}

As described in the reward setting in Section \ref{subsec:formation},
we set the regularization ratio $\prescript{1}{}{\alpha}=0.01$ for DP and $\prescript{2}{}{\alpha}=0$ for the order waiting time penalty as we don't consider the pick-up distance in the grid-based experiment.
Because of our homogeneous agent setting, all agents share the same $Q$-network and policy network for efficient training.
During the execution in the real-world environment, each agent can keep its copy of policy parameters and receive updates periodically from a parameter server.

\begin{table}
  \setlength{\tabcolsep}{2.8pt}
  \caption{Performance comparison regarding the normalized Gross Merchandise Volume (GMV) on the test set with respect to the performance of RAN.}
  \label{table:grid_gmv}
  \centering
  \begin{tabular}{cccc}
    \toprule
    $p_{sample}$ & 100\% & 50\% & 10\% \\
    \midrule
    RES
    & $-4.60 \pm 0.01$
    & $-4.65 \pm 0.04$
    & $-4.71 \pm 0.03$ \\
    REV
    & $+4.78 \pm 0.00$
    & $+4.87 \pm 0.02$
    & $+5.24 \pm 0.07$ \\
    IOD
    & $+6.27 \pm 0.01$
    & $+6.21 \pm 0.03$
    & $+5.73 \pm 0.07$ \\
    COD
    & $\textbf{+7.78} \pm \textbf{0.01}$
    & $\textbf{+7.76} \pm \textbf{0.03}$
    & $\textbf{+7.52} \pm \textbf{0.14}$ \\
    \bottomrule
  \end{tabular}
\end{table}

\subsubsection{Result Analysis}
\label{subsubsec:grid_result}
For all learning methods, we run $20$ episodes for training, store the trained model periodically, and conduct the evaluation on the stored model with the best training performance.
The training set is generated by bootstrapping 50\% of the original real orders unless specified otherwise.
We use five random seeds for testing and present the averaged result.
We compare the performance of different methods by three metrics, including the total income in a day (GMV), the order response rate (ORR), and the average order destination potential (ADP).
ORR is the number of orders taken divided by the number of orders generated, and ADP is the sum of destination potential (as described in Section \ref{subsec:formation}) of all orders divided by the number of orders taken.

\paragraph{Gross Merchandise Volume}
As shown in Table \ref{table:grid_gmv}, the performance of COD largely surpasses all rule-based methods and IOD in GMV metric.
RES suffers from lowest GMV among all methods due to its preference of short distance trips with lower average order value.
On the other hand, REV aims to pick higher value orders with longer trip time, thus enjoying a higher GMV.
However, both RES and REV cannot find a balance between getting higher income per order and taking more orders, while RAN falls into a sub-optimal trade-off without favoring either side.
Instead, our proposed MARL methods (IOD and COD) achieve higher growths in terms of GMV by considering each order's price and the destination potential concurrently.
Orders with relatively low destination potential will be less possible to get picked, thus avoiding harming GMV by preventing the driver from trapping in areas with very few future orders.
By direct modeling other agents' policies and capturing the interaction between agents in the environment, the COD algorithm with mean field approximation gives the best performance among all comparing methods.

\begin{table}
  \setlength{\tabcolsep}{2.8pt}
  \caption{Performance comparison regarding the order response rate (OOR) on the test set. The percentage difference shown for all methods is with respect to RAN.}
  \label{table:grid_oor}
  \centering
  \begin{tabular}{cccc}
    \toprule
    $p_{sample}$ & 100\% & 50\% & 10\% \\
    \midrule
    RES
    & $+3.15 \pm 0.01$ & $+3.11 \pm 0.02$ & $+2.59 \pm 0.04$ \\
    REV
    & $-6.21 \pm 0.00$ & $-6.14 \pm 0.03$ & $-5.22 \pm 0.06$ \\
    IOD
    & $+4.67 \pm 0.01$ & $+4.68 \pm 0.03$ & $+4.53 \pm 0.04$ \\
    COD
    & $\textbf{+5.35} \pm \textbf{0.00}$ & $\textbf{+5.38} \pm \textbf{0.03}$ & $\textbf{+5.42} \pm \textbf{0.09}$ \\
    \bottomrule
  \end{tabular}
\end{table}

\begin{figure}
  \centering
	\begin{subfigure}[b]{.49\linewidth}
		\centering
		\includegraphics[height=0.73\columnwidth]{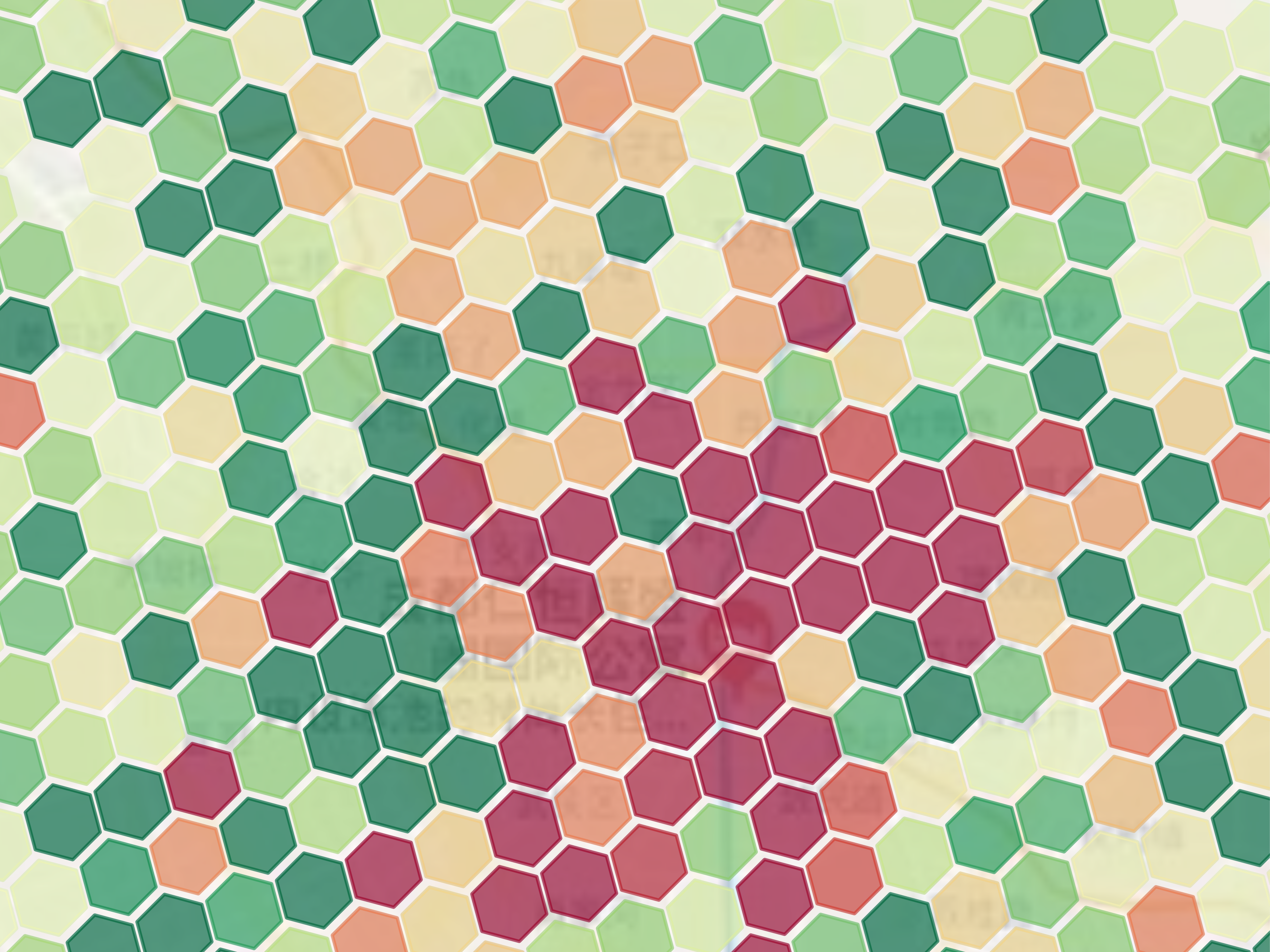}
		\caption{COD}
		\label{subfig:sd_cod}
	\end{subfigure}
	\begin{subfigure}[b]{.49\linewidth}
		\centering
		\includegraphics[height=0.73\columnwidth]{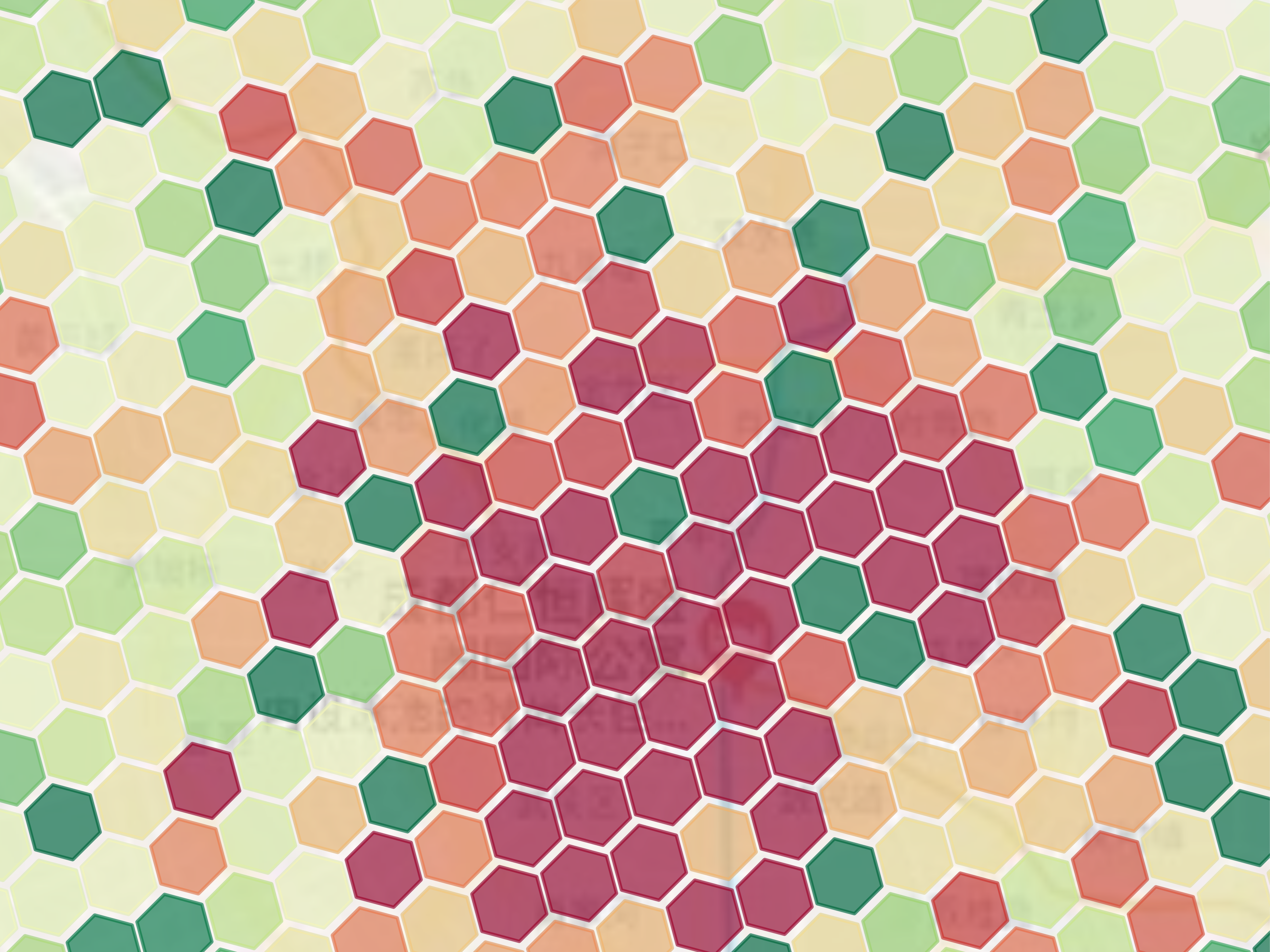}
		\caption{REV}
		\label{subfig:sd_rev}
	\end{subfigure}
  \caption{An example of the demand-supply gap in the city center during peak hours. Grids with more drivers are shown in green (in red if opposite) and the gap is proportional to the shade of colors.}
  \label{fig:sd}
\end{figure}

\paragraph{Order Response Rate}
In Table \ref{table:grid_oor} we compare the performance of different models in terms of the order response rate (OOR), which is the number of orders taken divided by the number of orders generated.
RES has a higher OOR than the random strategy as it focuses on reducing the trip distance to take more orders.
On the other hand, REV aims to pick higher value orders with longer trip time, leading to sacrifice on OOR.
Although REV has a relatively higher GMV than other two rule-based methods, its lower OOR indicates a lower customer satisfaction rate, thus failing to meet the requirement of an optimal order dispatching algorithm.
By considering both the average order value and the destination potential, IOD and COD achieve higher OOR as well.
High-priced orders with low destination potential, i.e., a long trip to a suburban area, will be less possible to get picked, thus avoiding harming OOR when trying to take a high-priced order.

\begin{table}
  \setlength{\tabcolsep}{2.8pt}
  \caption{Performance comparison in terms of the average destination potential (ADP) on the test set. The percentage difference shown for all methods is with respect to RAN.}
  \label{table:grid_adp}
  \centering
  \begin{tabular}{cccc}
    \toprule
    $p_{sample}$ & 100\% & 50\% & 10\% \\
    \midrule
    RES
    & $+6.36 \pm 0.15$ & $+7.40 \pm 0.54$ & $+9.97 \pm 0.43$ \\
    REV
    & $-20.55 \pm 0.16$ & $-20.23 \pm 0.52$ & $-20.55 \pm 1.58$ \\
    IOD
    & $+54.33 \pm 0.09$ & $+54.71 \pm 0.20$ & $+53.33 \pm 0.54$ \\
    COD
    & $\textbf{+66.74} \pm \textbf{0.05}$ & $\textbf{+66.90} \pm \textbf{0.29}$ & $\textbf{+64.69} \pm \textbf{0.61}$ \\
    \bottomrule
  \end{tabular}
\end{table}

\paragraph{Average Order Destination Potential}
To better present our method on optimizing the demand-supply gap, we list the average order destination potential (ADP) in Table \ref{table:grid_adp}.
Note that all ADP values are negative, which indicates the supply still cannot fully satisfy the demand on average.
However, IOD and COD largely alleviate the problem by dispatching orders to places with higher demand.
As shown in Fig. \ref{fig:sd}, COD largely fills the demand-supply gap in the city center during peak hours, while REV fails to assign drivers to take orders with high-demand destination, thus leaving many grids with unserved orders.

\begin{figure}[t]
  \centering
  \includegraphics[height=0.5\columnwidth]{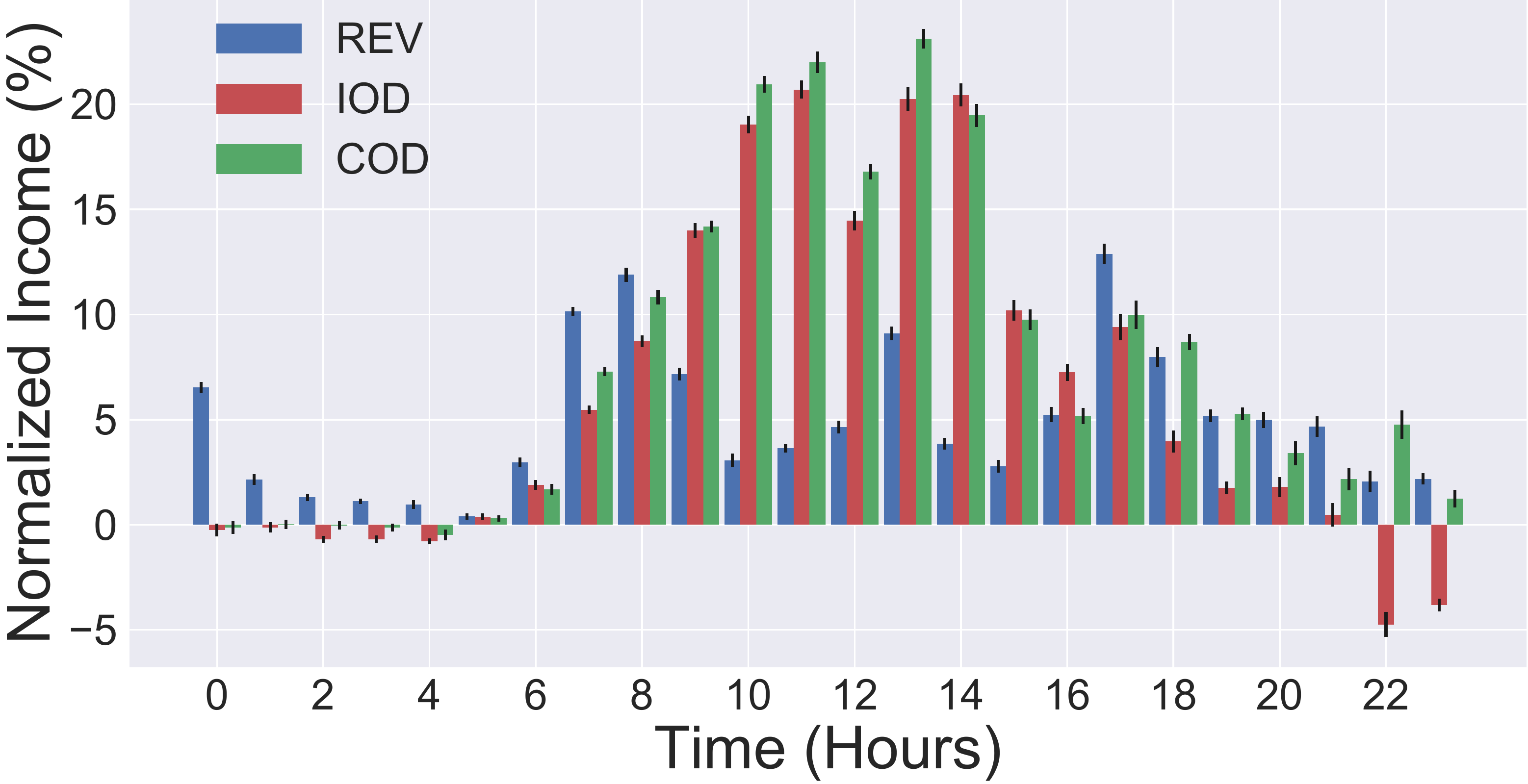}
  \caption{Normalized hourly income of REV, IOD and COD with respect to the average hourly income of RAN.}
  \label{fig:daycomparison_grid_all}
\end{figure}

\paragraph{Sequential Performance Analysis}
To investigate the performance change of our proposed algorithms regarding the time of the day, Fig. \ref{fig:daycomparison_grid_all} shows the normalized income of all algorithms with respect to the average hourly revenue of RAN. 
We eliminate RES from this comparison since it does not focus on getting a higher GMV and has relatively low performance.
A positive value of the bar graph illustrates an increase in income compared to RAN, and vice versa.
As shown in Fig. \ref{fig:daycomparison_grid_all}, MARL methods (IOD and COD) outperform RAN in most of the hours in a day (except for the late night period between 12 a.m. to 4 a.m., when very few orders are generated).
During the peak hours in the morning and at night, MARL methods achieve a much higher hourly income than RAN.
REV achieves a higher income than MARL methods between late night and early morning (12 a.m. to 8 a.m.) because of its aggressive strategy to take high-priced orders.
However, this strategy ignores the destination quality of orders and assigns many drivers with orders to places with low demand, resulting in a significant income drop during the rest of the day (except for the evening peak hours when the possibility of encountering high-priced orders is large enough for REV to counteract the influence of the low response rate).
On the other hand, IOD and COD earn significantly higher than RAN and REV for the rest of the day, possibly because of MARL methods' ability to recognize a better order in terms of both the order fee and the destination potential.
As the order choice of MARL methods will prevent the agent from choosing a destination with lower potential, agents following MARL methods are thus enjoying more sustainable incomes.
Also, COD outperforms IOD constantly, showing the effectiveness of explicitly conditioning on other agents' policies to capture the interaction between agents.

\begin{table}[t]
  \caption{Performance comparison of different reward settings applied to IOD with respect to RAN.}
  \label{table:rewardfunc_comparison}
  \centering
  \begin{tabular}{cccc}
    \toprule
    Method     & GMV (\%)    & OOR (\%) & ADP (\%)\\
    \midrule
    AVE
    & $-0.36 \pm 0.09$
    & $-0.14 \pm 0.07$
    & $+1.11 \pm 1.06$ \\
    IND
    & $+0.36 \pm 0.10$
    & $+0.30 \pm 0.09$
    & $-1.26 \pm 0.87$ \\
    AVE + DP
    & $-0.26 \pm 0.09$
    & $-0.10 \pm 0.06$
    & $+1.17 \pm 1.08$ \\
    IND + DP
    & $\textbf{+2.97} \pm \textbf{0.02}$
    & $\textbf{+2.54} \pm \textbf{0.03}$
    & $\textbf{+59.44} \pm \textbf{1.33}$ \\
    \bottomrule
  \end{tabular}
\end{table}

\paragraph{Effectiveness of Reward Settings}
As described in Section \ref{subsec:formation}, we use the destination potential as a reward function regularizer to encourage cooperation between agents.
To show the effectiveness of this reward setting, we compare the GMV, OOR, and ADP of setting average income (AVE) and  independent income (IND) as each agent's reward for IOD respectively.
We also measure the performance of adding DP as a regularizer for both settings.
For this experiment, we bootstrap 10\% of the original real orders for the training and test set separately.
As shown in Table \ref{table:rewardfunc_comparison}, the performance of AVE in terms of all metrics is relatively lower than those of IND methods and RAN (even with DP added).
This is possibly because of the credit assignment problem, where the agent's behavior is drowned by the noise of other agents' impact on the reward function.
On the other hand, setting the individual income as the reward helps to distinguish each agent's contribution to the global objective from others, while adding DP as a regularizer further encourages the coordination between agents by arranging them to places with higher demand.

\subsection{Coordinate-based Experiment}
\label{subsec:coord-based}

\subsubsection{Environment Setting}
As the real-world environment is coordinate-based rather than grid-based, we also conduct experiments on a more complex coordinate-based simulator provided by DiDi Chuxing.
At each time step $t$, the coordinate-based simulator provides an observation $\vect{o}^t$ including a set of active drivers and a set of available orders.
Each order feature includes the coordinate for the origin and the destination, while each driver has the coordinate as the location feature.
The order dispatching algorithm works the same as described in Section \ref{subsubsec:grid-based-setting}.
To better approximate the real-world scenario, this simulator also considers order cancelations, i.e., an order might be canceled during the pick-up process.
This dynamic is controlled by a random variable which is positively related to the arriving time.
The data resource of this simulator is based on historical dispatching events, including order generation events, driver logging on/off events and order fee estimation.
During the training stage, the simulator will load five weekdays data with $350K$ dispatching events and generate $9K$ drivers.
For evaluation, our model is applied on future days which are not used in the training phase.

\subsubsection{Model Setting}
We evaluate the performance of following MARL based methods including IOD, COD, and a DQN variation of IOD (Q-IOD), i.e., without the policy network.
We also compare these MARL methods with a centralized combinatorial optimization method based on the Hungarian algorithm (HOD).
The HOD method focuses on minimizing the average arriving time (AAT) by setting the weight of each driver-order pair with the pick-up distance.
For all MARL based methods, the same network architecture setting as described in Section 
\ref{grid-based-setting} is applied.
Except that we use a mini-batch size of 200 because of the shorter simulation gap.
The regularization ratio for pick-up distance is $\prescript{2}{}{\alpha}=-0.1$ in this experiment.

\subsubsection{Result Analysis}
We train all MARL methods for 400K iterations and apply the trained model in a test set (consists of three weekdays) for comparison. We compare different algorithms in terms of the total income in a day (GMV) and the average arriving time (AAT).
GMV2 considers the cancellation while GMV1 doesn't. All the above metrics are normalized with respect to the result of HOD.

\begin{table}[t]
  \caption{Performance comparison in terms of the GMV and the average arriving time (AAT) with respect to HOD.}
  \label{table:coord-result}
  \centering
  \begin{tabular}{cccc}
    \toprule
    Method     & GMV1     & GMV2  & AAT \\
    \midrule
    Q-IOD
    & $-0.02\%$
    & $-0.77\%$
    & $+6.76\%$ \\
    IOD
    & $+0.02\%$
    & $-0.48\%$
    & $+5.90\%$ \\
    COD
    & $\textbf{+0.32\%}$
    & $\textbf{+0.06\%}$
    & $+5.49\%$ \\
    \bottomrule
  \end{tabular}
\end{table}

As shown in Table \ref{table:coord-result}, the result of COD largely outperforms Q-IOD and IOD in both GMV1 and GMV2, showing the effectiveness of direct modeling of other agents' policies in MARL.
In addition, COD outperforms HOD in both GMV settings as well; this justifies the advantage of MARL algorithms that exploit the interaction between agents and the environment to maximize the cumulative reward.
The performance improvement of GMV2 is smaller than that of GMV1 for MARL methods. This is possibly because that HOD works by minimizing the global pick-up distance and has a shorter waiting time.
On the other hand, MARL methods only consider the pick-up distance as a regularization term, thus performing comparatively worse than HOD regarding AAT. 
As shown in Table \ref{table:coord-result}, the AAT of all MARL methods are relatively longer than that of the combinatorial optimization method.
However, as the absolute values of GMV are orders of magnitude higher than ATT, the increase in ATT is relatively minor and is thus tolerable in the order dispatching task.
Also, MARL methods require no centralized control during execution, thus making the order dispatching system more robust to potential hardware or connectivity failures. 

\section{Related Work}
\paragraph{Order Dispatching}
Several previous works addressed the order dispatching problem by either centralized or decentralized ruled-based approaches.
\citet{lee2004taxi} and \citet{10.1007/978-3-540-72590-9_96} chose the pick-up distance (or time) as the basic criterion, and focused on finding the nearest option from a set of homogeneous drivers for each order on a first-come, first-served basis.
These approaches only focus on the individual order pick-up distance; however, they do not account for the possibility of other orders in the waiting queue being more suitable for this driver.
To improve global performance, \citet{Zhang:2017:TOD:3097983.3098138} proposed a novel model based on centralized combinatorial optimization by concurrently matching multiple driver-order pairs within a short time window.
They considered each driver as heterogenous by taking the long-term behavior history and short-term interests into account.
The above methods work with centralized control, which is prone to the potential "single point of failure" \cite{lynch2009single}.

With the decentralized setting, \citet{seow2010collaborative} addressed the problem by grouping neighboring drivers and orders in a small multi-agent environment, and then simultaneously assigning orders to drivers within the group.
Drivers in a group are considered as agents who conduct negotiations by several rounds of collaborative reasoning to decide whether to exchange current order assignments or not.
This approach requires rounds of direct communications between agents, thus being limited to a local area with a small number of agents. 
\citet{alshamsi2009multiagent} proposed an adaptive approach for the multi-agent scheduling system to enable negotiations between agents (drivers) to re-schedule allocated orders.
They used a cycling transfer algorithm to evaluate each driver-order pair with multiple criteria, requiring a sophisticated design of feature selection and weighting scheme.

Different from rule-based approaches, which require additionally hand-crafted heuristics, we use a model-free RL agent to learn an optimal policy given the rewards and observations provided by the environment.
A very recent work by \citet{Xu:2018:LOD:3219819.3219824} proposed an RL-based dispatching algorithm to optimize resource utilization and user experience in a global and more farsighted view.
However, they formulated the problem with the single-agent setting, which is unable to model the complex interactions between drivers and orders.
On the contrary, our multi-agent setting follows the distributed nature of the peer-to-peer ridesharing problem, providing the dispatching system with the ability to capture the stochastic demand-supply dynamics in large-scale ridesharing scenarios.
During the execution stage, agents will behave under the learned policy independently, thus being more robust to potential hardware or connectivity failures.

\paragraph{Multi-Agent Reinforcement Learning}
One of the most straightforward approaches to adapt reinforcement learning in the multi-agent environment is to make each agent learn independently regardless of the other agents, such as independent $Q$-learning \cite{tan1993multi}. They, however, tend to fail in practice \cite{matignon2012independent} because of the non-stationary nature of the multi-agent environment.
Several approaches have been attempted to address this problem, including sharing the policy parameters \cite{gupta2017cooperative}, training the $Q$-function with other agent's policy parameters \cite{tesauro2004extending}, or using importance sampling to learn from data gathered in a different environment \cite{foerster2017stabilising}.
The idea of centralized training with decentralized execution has been investigated by several works \cite{foerster2017counterfactual, lowe2017multi, peng2017multiagent} recently for MARL using policy gradients \cite{sutton1999policy}, and deep neural networks function approximators, based on the actor-critic framework \cite{NIPS1999_1786}.
Agents within this paradigm learn a centralized $Q$-function augmented with actions of other agents as the \emph{critic} during training stage, and use the learned policy (the \emph{actor}) with local observations to guide their behaviors during execution.
Most of these approaches limit their work to a small number of agents usually less than ten.
To address the problem of the increasing input space and accumulated exploratory noises of other agents in large-scale MARL, \citet{pmlr-v80-yang18d} proposed a novel method by integrating MARL with mean field approximations and proved its convergence in a tabular $Q$-function setting.
In this work, we further develop  MFRL and prove its convergence when the $Q$-function is represented by function approximators.

\section{Conclusion}
In this paper, we proposed the multi-agent reinforcement learning solution to the order dispatching problem. Results on  two large-scale simulation environments have shown that our proposed algorithms (COD and IOD) achieved (1) a higher GMV and OOR than three rule-based methods (RAN, RES, REV);
(2) a higher GMV than the combinatorial optimization method (HOD), with desirable properties of fully distributed execution;
(3) lower supply-demand gap during the rush hours, which indicates the ability to reduce traffic congestion.
We also provide the convergence proof of applying mean field theory to MARL with function approximations as the theoretical justification of our proposed algorithms.
Furthermore, our MARL approaches could achieve fully decentralized execution by distributing the centralized trained policy to each vehicle through Vehicle-to-Network (V2N).
For future work, we are working towards controlling ATT while maximizing the GMV with the proposed MARL framework. Another interesting and practical direction to develop is to use a heterogeneous agent setting with individual specific features, such as the personal preference and the distance from its own destination.

\bibliographystyle{ACM-Reference-Format}
\small
\bibliography{didi}

\end{document}